\documentclass{article}
\PassOptionsToPackage{numbers,sort&compress}{natbib}
\usepackage[final]{neurips_2024}

\usepackage[utf8]{inputenc}
\usepackage[T1]{fontenc}
\usepackage{amsfonts}
\usepackage{amsmath}
\usepackage{amssymb}
\usepackage{amsthm}
\usepackage{thmtools}
\usepackage{thm-restate}
\usepackage{url}
\usepackage{booktabs}
\usepackage{graphicx}
\usepackage{microtype}
\usepackage{xcolor}
\usepackage[ruled,vlined,linesnumbered]{algorithm2e}
\usepackage{braket}
\usepackage{mathtools}
\usepackage{caption}
\usepackage{subcaption}
\usepackage{hyperref}
\usepackage[nameinlink,capitalise,noabbrev]{cleveref}
\usepackage{enumitem}

\declaretheorem{theorem}
\declaretheorem[sibling=theorem]{proposition}
\declaretheorem[sibling=theorem]{lemma}
\declaretheorem[parent=theorem]{corollary}
\declaretheorem[style=remark, parent=theorem]{remark}
\declaretheorem[style=definition]{definition}
\declaretheorem[style=remark, parent=definition]{note}

\DeclareMathOperator{\poly}{poly}
\DeclareMathOperator*{\argmin}{arg\,min}
\DeclareMathOperator*{\argmax}{arg\,max}

\IncMargin{1.5em}

\mathchardef\mhyphen="2D

\allowdisplaybreaks
\hypersetup{colorlinks,
    citecolor=blue,
    linkcolor=red,
    urlcolor=purple,
    breaklinks=true}
    
\newcommand{\ketbra}[2]{\lvert#1\rangle\langle#2\rvert}
\newcommand{\bm}[1]{\boldsymbol{#1}}

\title{Quantum algorithm for large-scale market equilibrium computation}

\author{
  Po-Wei Huang\textsuperscript{1} and Patrick Rebentrost\textsuperscript{1, 2}\\
  \textsuperscript{1}Centre for Quantum Technologies,  National University of Singapore, Singapore 117543\\
  \textsuperscript{2}Department of Computer Science, National University of Singapore, Singapore 117417\\
  \texttt{huangpowei22@u.nus.edu},  \texttt{patrick@comp.nus.edu.sg}
}

\begin{document}

\maketitle
\begin{abstract}
Classical algorithms for market equilibrium computation such as proportional response dynamics face scalability issues with Internet-based applications such as auctions, recommender systems, and fair division, despite having an almost linear runtime in terms of the product of buyers and goods.  
In this work, we provide the first quantum algorithm for market equilibrium computation with sub-linear performance. Our algorithm provides a polynomial runtime speedup in terms of the product of the number of buyers and goods while reaching the same optimization objective value as the classical algorithm. Numerical simulations of a system with 16384 buyers and goods support our theoretical results that our quantum algorithm provides a significant speedup.
\end{abstract}

\section{Introduction}
The balance of supply and demand is a fundamental and well-known law that determines the price of goods in a market. In a market with a set of $n$ buyers and $m$ goods, the \emph{competitive equilibrium}~\citep{arrow1951extension, debreu1951coefficient} determines the optimal price and allocation of goods such that the supply equals the demand in the given market. The computation of the competitive equilibrium is known as the \emph{market equilibrium computation problem}, whose unique solution was shown to exist under a general model of the economics in the seminal work of \citet{arrow1954existence}. The relevance of such problems in algorithmic game theory~\citep{vazirani2007combinatorial, codenotti2007computation} is substantiated by the first welfare theorem, which implies that the competitive equilibria are \emph{Pareto-efficient}~\citep{mascolell1995equilibrium}, where no allocation is available that makes one agent better without making another one worse. In \emph{competitive equilibrium from equal income} (CEEI) scenarios, such equilibria are further known to by \emph{envy-free}~\citep{foley1967resource, varian1974equity}, where no agent would prefer an allocation received by another agent over their own.

The market equilibrium computation problem has, in recent years, been extended to various large-scale Internet-based markets~\citep{kroer2021market}, including auction markets~\citep{conitzer2022pacing}, fair item allocation/fair division~\citep{othman2010finding, budish2017course, babaioff2019fair}, scheduling problems~\citep{im2017competitive} and recommender systems~\citep{kroer2022computing}. Such developments call for the need to further develop algorithmic theories for markets and the computation of market equilibria. We focus on a particular type of market known as the Fisher market~\citep{fisher1891mathematical, brainard2005how}, where there is a set of $n$ buyers interested in buying $m$ infinitely-divisible goods, and where each buyer has their monetary budget that has no intrinsic value apart from being used to purchase goods. We mainly consider Fisher markets with linear utilities, where the total utility gained by purchasing goods is strictly linear to the value and proportion of the goods obtained. 

For combinatorial formulations of the market equilibrium problem, algorithms have been discovered that can obtain exact and approximate solutions~\citep{scarf1967core, devanur2008market, orlin2010improved, vegh2012strongly}, these algorithms scale poorly against the growing number of buyers and goods. One can also formulate the market equilibrium computation problem as an optimization problem that maximizes a convex objective function known as the Eisenberg-Gale (EG) convex program~\citep{eisenberg1959consensus, eisenberg1961aggregation}. For such optimization problems, approximate solutions can be found much faster. One such algorithm of the market equilibrium problem is the \emph{proportional response} (PR) dynamics~\citep{wu2007proportional, zhang2011proportional}. The PR dynamics is an iterative algorithm that converges with a rate of $\frac{1}{T}$ where $T$ is the number of iterations. Each iteration has a cost of $\mathcal{O}(mn)$ from proportionally updating individual bids of each buyer for different goods.

Given the high number of buyers and goods that can exist in Internet-based markets, the problem of further algorithmic speedups to the computation continues to be an active field of research. \citet{gao2020first} discovered that by using projected gradient descent instead of PR dynamics, the market equilibrium can be found with linear convergence. Other attempts that aim to reduce the cost per iteration, such as using clustering to reduce the problem size~\citep{kroer2022computing}, have also been made. However, it is unclear whether advantages beyond a constant-factor speedup can be provided.

While it is unclear whether additional classical strategies can provide further algorithmic speedup to the market equilibrium computation problem, one can utilize algorithmic developments in quantum computation to achieve such goals. Quantum computation~\citep{nielsen2010quantum} is an emerging technology that has been utilized for algorithmic speedups in various optimization problems~\citep{abbas2024challenges} such as linear programming~\citep{vanapeldoorn2019quantum} and semidefinite programming~\citep{brandao2017quantum,brandao2019quantum,vanapeldoorn2019improvements,vanapeldoorn2020quantum}.

In this work, we consider a Fisher market with $n$ buyers and $m$ goods, where the objective is to find an approximate market equilibrium whose EG objective function is within an additive error $\varepsilon$ of the optimal EG objective value. We provide a method to reduce the cost per iteration by utilizing quantum norm estimation and quantum inner product estimation~\citep{li2019sublinear, rebentrost2021quantum} and provide the first quantum algorithm to achieve sublinear performance in terms of the product of the buyers and goods in market equilibrium computation. To arrive at the quantum algorithm, we show an alternate version of the PR dynamics with erroneous updates, which we term the faulty proportional response (FPR) dynamics. We then provide a quantum algorithm that provides a quadratic speedup in terms of the smaller dimension between buyers and goods, as well as less memory consumption, albeit being based on QRAM instead of classical RAM. We summarize our results in \cref{tabMain}. 

\begin{table}
\centering
\begin{tabular}{@{}lccccc@{}}
\toprule
Algorithm  & Iterations &  Runtime & Memory & Result Prep.\\
\midrule
PR dynamics~\citep{wu2007proportional}  &  $\dfrac{\log m}{\varepsilon}$ & $\tilde{\mathcal O} \Big(\dfrac{mn}{\varepsilon}\Big)$ & $\mathcal O (mn)$ & N/A, in RAM\\
Our work & $\dfrac{2\log m}{\varepsilon}$ & $ \tilde{\mathcal O} \bigg(\dfrac{\sqrt{mn\max(m,n)}}{\varepsilon^2}\bigg)$ & $\mathcal O (m+n)^*$ & \begin{tabular}{@{}c@{}}QA: $\mathcal O (\poly\log \frac{mn}{\varepsilon})$ \\ SA: $\tilde{\mathcal O} (\sqrt{mn})$\end{tabular}\\
\bottomrule
\end{tabular}
\caption{\textit{Main results.} In this work, $n$ is the number of buyers, $m$ is the number of goods, and $\varepsilon$ indicates the additive error of the computed values to the minimally-achievable EG objective value. The memory complexity for the quantum algorithm (annotated with *) refers to the use of quantum query access to classical memory, achievable by QRAM (see \cref{defQRAM}), instead of classical RAM. As the computed competitive equilibrium consumes $\mathcal{O}(mn)$ memory, our quantum algorithm does not provide the entire bid matrix, but instead provides quantum query access (QA) and sample access (SA) to the competitive equilibrium. The result preparation column refers to the additional runtime cost of preparing QA and SA.}
\label{tabMain}
\end{table}

Unlike the classical PR algorithm, which provides an entire matrix corresponding to the competitive equilibrium with storage $\mathcal O (mn)$ space,\footnote{Given that it takes $\mathcal O (mn)$ to even output the entire matrix, a sublinear algorithm would require a different output model in the first place, since the sublinearity would be undone by the solution generation.} we provide quantum query and sample access to the values of the competitive equilibrium, which allows access to the individual values by index querying as well as $\ell_1$ sample access to the values. This access format has previously been used as the output model for quantum recommendation systems~\citep{kerenidis2017quantum}, quantum linear system solvers~\citep{harrow2009quantum}, but has been pointed out to have significant caveats~\citep{aaronson2015read}. While the quantum algorithm does not exactly ``solve'' the market equilibrium problem in the sense of outputting the entire matrix, quantum query and sample access provide a preparation of a quantum state encoding the solution that can be used for further computation, for extracting only a small set of values within the matrix, or for extracting certain properties of the matrix, which may indeed be the use case for large-scale distributed systems.

\section{Preliminaries}

\paragraph{Notations.} Let $[n]:=\{0, 1, \dots, n-1\}$. We use $\odot$ to represent element-wise multiplication, as well as $\oplus$ for bit-wise XOR operation and $\otimes$ for tensor products. For vectors $v \in \mathbb{R}^N$, we denote a vector's $\ell_p$ norm by $\|v\|_p:= \sqrt[p]{\sum_{i=1}^N |v_i|^p}$. Let $\mathcal M_{M\times N}(\mathbb{R})$ indicate the space of square matrices of size $M\times N$ over $\mathbb{R}$. We denote the $i$-th row vector of $A$ by $A_{i, *}$ and the $j$-th column vector of $A$ by $A_{*, j}$. We further define $\mathbb{I}$ as $[0,1]$, and the $n$-unit simplex as $\mathbb{S}^n$, i.e. $\mathbb{S}^n = \{v \in \mathbb{I}^n, \|v\|_1 = 1\}$. For sets of numbers, we add the subscript $\cdot_+$ to indicate a constraint on positivity for elements in the set. We use $\ket{k}$ to denote a binary encoding of a real number $k$ up to arbitrary precision into a quantum state, and $\ket{\bar{0}}$ to denote a multi-qubit zero state whose number of qubits can be inferred from the context. Lastly, we use $\mathcal {\tilde O} (\cdot)$ to omit polylogarithmic factors in asymptotic runtime/memory analysis.

\paragraph{Quantum computation.}
Quantum algorithms are shown to be able to provide asymptotic speedups over classical counterparts~\citep{montanaro2016quantum,dalzell2023quantum, abbas2024challenges} by utilizing characteristics of quantum mechanics such as preparing superpositions of different computational paths.

This property of quantum algorithms allows one to simulate the probability distribution of classical randomized algorithms directly as amplitudes of the quantum state, with the accepted results of the classical algorithms labelled as ``good'' states, and the rejected results as ``bad'' states. The \emph{quantum amplitude amplification} (QAA)~\citep{brassard2002quantum} technique amplifies the amplitudes ``good'' states such that the success probability increases to a sufficiently high occurrence from its original $p$ upon measurement at a runtime cost of $\mathcal{O}(1/\sqrt{p})$. This provides a quadratic speedup compared to classical Monte Carlo methods, which take on average $\mathcal{O}(1/p)$ samples to achieve success. The QAA technique gives way to another technique known as \emph{quantum amplitude estimation} (QAE)~\citep{brassard2002quantum}, which combines QAA with eigenvalue estimation via quantum phase estimation (QPE)~\citep{kitaev1995quantum} to directly estimate the probability value of the occurrence of a certain event with a quadratic speedup.

\begin{theorem}[Quantum amplitude estimation; Theorem 12, \citep{brassard2002quantum}, Formulation of \citep{montanaro2015quantum}]\label{theoremQAE}
Let $t\in \mathbb{N}$. We are given one copy of a quantum state $\ket \psi$ as input, as well as a unitary transformation $U = I - 2\ketbra{\psi}{\psi}$, and a unitary transformation $V = I - 2P$ for some projector $P$. There exists a quantum algorithm that outputs $\tilde a$, an estimate of $a = \|P\ket{\psi}\|^2$, such that
\begin{equation*}
    \textstyle|\tilde a - a| \le 2\pi\frac{\sqrt{a(1-a)}}{M}+\frac{\pi^2}{M^2}
\end{equation*}
with probability at least $8/\pi^2$, using $M$ applications of $U$ and $V$ each.
\end{theorem}
In this paper, we use QAE to estimate $\ell_1$ norms and inner products of vectors $v\in\mathbb{I}^N$ up to a multiplicative error in $M \in \mathcal{O}(\frac{\sqrt{N}}{\varepsilon}\log(\frac{1}{\delta}))$ runtime with probability $1-\delta$~\citep{rebentrost2021quantum}, invoking a quadratic speedup in both the dimension and the error rate. We defer the formulation and details to \cref{appendixQroutine}.

Apart from quantum subroutines that provide speedups, we also require the usage of arithmetic operations such as addition, subtraction, multiplication, and division on quantum computers. We assume the arithmetic model, which would allow us to ignore issues arising from the fixed point representation of numbers.\footnote{If the fixed point representation with an additive error of $\mu$ is considered, the additional multiplicative cost required for operations is then $\mathcal{O}(\poly\log\frac{1}{\mu})$. Considering $\mu \in \Omega(1/\poly (m,n))$, the additional cost is $\mathcal{O}(\poly\log (m,n))$, which are polylogarithmic factors that we already omit in this paper.} We further assume that we have access to quantum arithmetic circuits~\citep{vedral1996quantum, takahashi2009quantum} that can perform such arithmetic operations in $\mathcal O(1)$ gates, and that by using such circuits, computation of the $n$-th power of a number, where $n\in\mathbb{N}$, can be achieved in $\mathcal{O}(\poly \log n)$ gates, using methods like binary exponentiation~\citep{montgomery1987speeding}. We note that quantum arithmetic circuits can be used to execute the same operation on multiple numbers in parallel if the numbers are held in superposition.

In regards to the access of data encoding in quantum states, there are two main models -- quantum query access and quantum sample access. For clarity, we highlight quantum states that indicate the index in such access models in bold font throughout the manuscript.

\begin{definition}[Quantum query access]
Let $n \in \mathbb N$ and $c\in \mathcal O (1)$. Let a vector of bit strings $w$ be such that $\forall j \in [n], w_j \in \{0, 1\}^c$, and let an arbitrary bit string be $x\in \{0, 1\}^c$. We define quantum query access to $w$ as the access to individual bit strings in $w$ for $j \in [n]$ in the format of
\begin{equation*}
\bm{\ket{j}}\ket{x} \rightarrow \bm{\ket{j}}\ket{x \oplus w_j}
\end{equation*}
operating on $\mathcal O (\log n)$ qubits.
\end{definition}
\begin{note}
Note that when $x = 0$, the quantum register stores $\ket{w_j}$ after the query. When $x = w_j$, the quantum register stores $\ket{0}$ after the query. Hence, two consecutive queries onto a quantum register is the identity operation.
\end{note}
\begin{definition}[Quantum sample access]
Let $n \in \mathbb N$ and $w\in \mathbb R^n$. We define quantum sample access to $w$ as the access to the index $j\in[n]$ by probability $w_j/\|w\|_1$ in the format of
\begin{equation*}
\textstyle\ket{\bar0} \rightarrow \sum_{j=0}^{n-1}\sqrt{\frac{w_j}{\|w\|_1}}\bm{\ket{j}}
\end{equation*}
operating on $\mathcal O (\log n)$ qubits.
\end{definition}
\begin{note}
In quantum mechanics, Born's rule~\citep{born1926zur} postulates that the probability of measurement outcome corresponding to a state $\bm{\ket{j}}$ is proportional to the square of its amplitude under superposition. Hence, the square root in the amplitude regarding the format of quantum sample access.
\end{note}

Lastly, we need to access the input matrices and intermediate vectors as a superposition of encoded quantum states. Such quantum query access to classical data in memory can be achieved by \emph{quantum random access memory} (QRAM) as follows.\footnote{Our memory unit can be more precisely termed QRACM~\citep{kuperberg2013another, jaques2023qram}, including a classical memory, or quantum read-only memory (QROM)~\cite{babbush2018encoding} as opposed to QRAQM~\citep{kuperberg2013another, jaques2023qram} or the quantum random access gate (QRAG)~\citep{ambainis2007quantum}, whose memory registers store quantum states instead of classical numbers. However, both are more commonly and jointly referred to as QRAM in the literature.} We refer the reader to \citep{jaques2023qram} for a detailed survey.
\begin{definition}[Quantum random access memory;~\citep{giovannetti2008quantum, giovannetti2008architectures}]
    \label{defQRAM}
    Let $n \in \mathbb N$ and $c\in \mathcal{O}(1)$. Also let $w$ be a vector of bit strings such that $\forall i \in [n], w_i \in \{0, 1\}^c$. A quantum RAM provides quantum query access to $w_i$ in superposition after a one-time construction cost of $\tilde{\mathcal O}(n)$, where each access costs $\mathcal O(\poly\log n)$.
\end{definition}

\begin{note}
Quantum sample access can also be provided by QRAM via Grover-Rudolph procedure~\citep{grover2002creating} during construction with logarithmic overhead but is not used in this work.
\end{note}

\paragraph{Fisher market equilibrium.}
In the Fisher market model~\citep{fisher1891mathematical, brainard2005how}, we are given a market of $m$ infinitely divisible goods to be divided among $n$ buyers. Without loss of generality, we assume a unit supply for each good. Each buyer $i \in [n]$ has a budget of $B_i>0$ that has no intrinsic value apart from being used to purchase goods where, again without loss of generality, we assume that the corresponding vector $B\in \mathbb{S}^n$. Each buyer also has a utility function $u_i: \mathbb{R}^m \to \mathbb{R}_+$ that maps an allocation of portions of $m$ items to a utility value. We can then define the allocation matrix $x \in \mathcal{M}_{n\times m}(\mathbb{I})$ such that $x_{ij}$ is the portion of item $j$ allocated to buyer $i$, where $x_i \in \mathbb{I}^m$ is the bundle of products allocated to buyer $i$. In this paper, we consider linear utility functions such that $u_i(x_i) = \sum_{j \in [m]} v_{ij}x_{ij}$, where $v_{ij} > 0$ is the value for a unit of item $j$ for buyer $i$.

Given the Fisher market, we want to compute its {\it competitive equilibrium}, which consists of the price vector $p\in\mathbb{R}_+^m$ for each item $j$ and allocation matrix $x$ such that each buyer $i$ exhausts their entire budget $B_i$ to acquire a bundle of items $x_i$ that maximizes each of their utility $u_i(x_i)$.

The market equilibrium of Fisher markets can be captured by solving the Eisenberg-Gale (EG) convex program~\citep{eisenberg1959consensus, eisenberg1961aggregation}. The program is derived from maximizing the budget-weighted geometric mean of the buyers' utilities i.e. the Nash social welfare, satisfying natural properties such as invariance of the optimal solution to rescaling and splitting \citep{jain2010eisenberg}, and balances the efficiency and fairness regarding the allocation of goods. Applying the $\log$ on the geometric mean, the EG program is as follows: 
\begin{equation}
     \textstyle\max_{x\ge0}  \sum_{i \in [n]} B_i\log u_i(x_i) \text{ s.t. }\sum_{i \in [n]} x_{ij} = 1, \forall j \in [m].
\end{equation}
where the price $p_j$ is the dual variable of the constraint on $x_{ij}$. Such convex programs (maximization of a concave function subject to constraints) can be solved by interior point methods~\citep{ boyd2004interior}, but may not scale to large markets. We discuss this further in \cref{sectDiscussion}.

For the linear Fisher market, an alternative convex program that obtains the same market equilibrium was shown by \citet{shmyrev2009algorithm}. Supposing that each buyer $i$ submits a bid $b_{ij}$ for item $j$ such that the sum of the bid of the buyer matches their budget $B_i$ such that each buyer $i$ is allocated $x_{ij} = b_{ij}/p_j$ of item $j$, we have the following convex program: 
\begin{equation}
\textstyle\max_{b\ge0}\sum_{ij}b_{ij}\log \frac{v_{ij}}{p_j} \text{ s.t. }\sum_{i \in [n]} b_{ij} = p_j, \forall j \in [m]; \sum_{j \in [m]} b_{ij} = B_i, \forall i \in [n].
\end{equation}
As the allocation matrix and price vector can be directly computed from and conversely, be used to compute the bid matrix, the bid matrix can be used as a direct representation of the market equilibrium itself, and hence, is the output of the algorithm we discuss in our paper.

\paragraph{Proportional response dynamics.}
The proportional response (PR) dynamics is an iterative algorithm~\citep{wu2007proportional,zhang2011proportional,levin2008bittorrent} that obtains the Fisher market equilibrium computation by updating the bids $b_{ij}$ submitted by buyer $i$ for item $j$. For each time step $t$, the elements of the price vector $p_j^{(t)}$ are computed by summing the bids for item $j$ such that $p_j^{(t)} = \sum_i b_{ij}^{(t)}$. The allocation $x_{ij}^{(t)}$ is then obtained by taking $x_{ij}^{(t)} = b_{ij}^{(t)} / p_j^{(t)}$. The buyers then update the bids such that the new bid is proportional to the utility $u_i^{(t)}= \sum_{j}v_{ij}x_{ij}^{(t)}$ gained in the current time step such that $b_{ij}^{(t+1)} = B_i v_{ij}x_{ij}^{(t)}/u_i^{(t)}$.
It was shown by \citet{birnbaum2011distributed} that the PR dynamics is equivalent to mirror descent~\citep{nemirovsky1983problem,beck2003mirror} with respect to a Bregman divergence~\citep{bregman1967relaxation} of the Shmyrev convex program.

For ease of discussion, we write the objective function of the EG and Shmyrev convex programs as functions of the bid matrix $b$, obtaining the EG objective function  $\Phi(b) = -\sum_{i \in [n]} B_i\log u_i$ and Shmyrev objective function  $\Psi(b) = \sum_{i \in [n], j \in [m]}  b_{ij}\log \frac{p_j}{v_{ij}}$.
We denote the optimal bid $b^* = \argmin_{b\in \mathcal{S}} \Phi(b)$, where $\mathcal S = \big\{b \in \mathcal M_{n\times m}(\mathbb I): \sum_{j\in[m]} b_{ij} = B_i\big\}$.

The convergence bounds of the PR dynamics regarding the EG and Shmyrev objective functions for linear Fisher markets were found as follows:
\begin{theorem}[Convergence of PR dynamics; \citep{birnbaum2011distributed}]\label{theoremPRconverge}
    Considering a linear Fisher market, for $b_{ij}^{(t)}$ as iteratively defined by the proportional response dynamics where $b_{ij}^{(0)} = \frac{B_i}{m}$, we have 
    \begin{equation}
        \textstyle\Psi(b^{(T)}) - \Psi(b^*) \le \frac{\log m}{T}, \quad \Phi(b^{(T-1)}) - \Phi(b^*) \le \frac{\log m}{T}.
    \end{equation}
\end{theorem}
We provide an alternate end-to-end proof of the convergence of both convex programs in \cref{appendixPRProof} that varies from \citet{birnbaum2011distributed}'s approach and, unlike the latter, is centered around the EG function instead of the Shmyrev function. Elements of this proof are used in the proof of theorems in later sections. Two notable results that we prove and utilize are:  1) $\Psi(b^{(t+1)}) \le \Phi(b^{(t)}) + \sum_{i \in [n]} B_i \log B_i \le \Psi(b^{(t)})$, and 2) the telescoping sum of the difference of the KL divergence~\citep{kullback1951information} of the optimal bid and the iterating bids can be lower bounded by the difference of the current EG objective function and the optimal EG function.

\section{Faulty proportional response dynamics}
Before moving on to our quantum algorithm, we propose the faulty proportional response (FPR) dynamics, which computes an erroneous update to compute a sequence of bids $\hat b_{ij}^{(t)}$. Such updates still retain a convergence guarantee, and serve as a counterpart to \cref{theoremPRconverge}. We first define a faulty update we use for the FPR dynamics:
\begin{definition}[Faulty proportional response update]\label{defFPRupdate}
    Let $t\ge 0$ and $\hat b^{(t)} \in \mathcal M_{n\times m}(\mathbb I_+)$. Let us be given $\varepsilon_p \in (0, 0.5)$ and $\tilde p^{(t)}$ such that $\forall j, t, |\tilde p_j^{(t)} - \hat p_j^{(t)}|\le  \hat p_j^{(t)}\varepsilon_p$ where $\hat p_j^{(t)} = \sum_{i\in [n]} \hat b_{ij}^{(t)}$. 
    Further, let us be given $\varepsilon_\nu \in (0, 0.5)$ and $\tilde \nu^{(t)}$ such that $\forall i, t, |\tilde \nu_i^{(t)} - \hat \nu_i^{(t)}|\le  \hat \nu_j^{(t)}\varepsilon_\nu$ where $\hat \nu_i^{(t)} = \sum_{j\in [m]} v_{ij} \hat b_{ij}^{(t)}/\tilde p_j^{(t)}$. A faulty proportional response update of the bids from timestep $t$ to $t+1$ is then expressed as follows:
    \begin{equation*}
    \textstyle\hat x_{ij}^{(t)} = \frac{\hat b_{ij}^{(t)}}{\tilde p_{j}^{(t)}}, \quad \hat b_{ij}^{(t+1)} = B_i \frac{v_{ij}\hat x_{ij}^{(t)}}{\tilde \nu_i^{(t)}}.
    \end{equation*}
\end{definition}
Note that while $\tilde p_j$ provides an estimation to the price $\hat p_j = \sum_{i\in[n]} \hat b_{ij}$, $\tilde \nu_i$ does not provide an estimation to the exact utility $\hat u_i = \sum_{j\in[m]} v_{ij}\hat b_{ij}/\hat p_j$. Instead, $\tilde \nu_i$ estimates $\hat \nu_i$, which replaces $\hat p_j$ in the computation of $u_i$ with $\tilde p_j$.

We find the convergence bounds of the FPR dynamics regarding the EG objective function for linear Fisher markets are as follows:
\begin{restatable}[Convergence of the FPR dynamics]{theorem}{fprconverge}
\label{theoremFPRconverge}
    Considering a linear Fisher market, for $b_{ij}^{(t)}$ as iteratively defined by the faulty proportional response dynamics where $\hat b_{ij}^{(0)} = \frac{B_i}{m}$, we have 
    \begin{equation*}
    \textstyle\min_{t\in [T]}\Phi(\hat b^{(t)}) - \Phi(b^*) \le \frac{2\log m}{T}
    \end{equation*}
    when $\varepsilon_\nu \le \frac{\log m}{8T}$ and $\varepsilon_p \le \frac{\log m}{6T}$.
\end{restatable}
A high-level idea of the proof follows from the telescoping sum trick to upper bound the EG objective functions with KL divergence from our proof of PR dynamics but with the consideration of error. We show an end-to-end proof of the convergence of the EG objective function in \cref{appendixFPRProof}.

Notice that in the FPR dynamics, we do not enforce the monotonicity of the iterations, but instead simply take the minimum value over all iterations. The error terms $\varepsilon_p$ and $\varepsilon_\nu$ in the FPR dynamics are only upper bounded such that the total sum of objective values over $T$ iterations (plus the original iteration) can be upper bounded by $\log m$ plus an accumulated error over $T$ iterations also within $\log m$. If we enforce the monotonicity of the iterations to take the last iteration, the error would require $\mathcal{O}(1/T^2)$ precision and would incur a further multiplicative overhead of $T$ in our algorithm.

However, given the formulation of a faulty update, a problem that comes into question is whether the computation of the exact value of the function $\Phi(b)$ is supported, as we do not compute $u_i$ in the process of updating. Without computation of $\Phi(b)$, one can not be sure which iteration of $\hat b^{(t)}$ is the minimum. On the other hand, we can use the computed value of $\tilde\nu_i^{(t)}$ as an estimator for the function $\Phi(b)$. The following result is then obtained.
\begin{restatable}{theorem}{esticonverge}
\label{theoremEsticonverge}
    Considering a linear Fisher market, for $b_{ij}^{(t)}$ as iteratively defined by the faulty proportional response dynamics where $\hat b_{ij}^{(0)} = \frac{B_i}{m}$. Let $t^* = \argmax_{t\in[T]} \sum_{i\in[n]} B_i \log \tilde\nu_i^{(t)}$. Then
    \begin{equation*}
    \textstyle \Phi(\hat b^{(t^*)}) - \Phi(b^*) \le \frac{2\log m}{T}
    \end{equation*}
    when $\varepsilon_\nu \le \frac{\log m}{8T}$ and $\varepsilon_p \le \frac{\log m}{6T}$. 
\end{restatable}
The proof of this theorem can similarly be found in \cref{appendixFPRProof}, which has the same proof idea as \cref{theoremFPRconverge} apart from some slight differences in error handling.

\section{Quantum algorithm}
We present our quantum algorithm for solving linear Fisher market equilibrium computation based on the FPR dynamics. Our quantum algorithm does not aim to provide speedups in terms of the number of iterations but provides speedups on the iteration cost of the PR dynamics algorithm. Our algorithm, while reducing the runtime in terms of the number of buyers $n$ or goods $m$, increases runtime in terms of the number of iterations $T$ but as the $T$ is logarithmically dependent on $m$, there is an overall quadratic speedup provided in the smaller of the two dimensions.

In this section, we further assume that $v \in \mathcal M_{n\times m} (\mathbb{I}_+)$. We note that the multiplicative scaling of $v_{ij}$ does not affect the bid matrix $b$ generated in the FPR dynamics as errors are multiplicative. Hence if the values are larger than $1$, we scale down the values by dividing the queried $v_{ij}$ by a number that is larger than $\max_{ij} v_{ij}$.

To compute the market equilibrium for the Fisher market by the FPR dynamics in the quantum setting, we require the data input of both the budget vector $B \in \mathbb{S}^n$ and the value matrix $v \in \mathcal M_{n\times m} (\mathbb{I}_+)$. We assume quantum query access to the budget and vector and value matrix by the index is readily given to us as part of the problem input without having to load classical data into a quantum system. That is, given an index state and ancilla quantum registers we can store the value of the budget and value according to the index in the ancilla register. Note that these operations can be performed in superposition, such that $\frac{1}{\sqrt{mn}}\sum_{i, j}\bm{\ket{i}}\bm{\ket{j}}\ket{\bar{0}} \rightarrow \frac{1}{\sqrt{mn}}\sum_{i, j}\bm{\ket{i}}\bm{\ket{j}}\ket{v_{ij}}$.  

We do not explicitly state how the data input of the budget and value entries are generated; they could be extracted from entries of a matrix already preloaded in QRAM, or
generated/reconstructed from a low-rank approximation of the matrix~\citep{kroer2022computing}, which takes $\mathcal O(k\poly\log mn)$ cost to access $k$-rank approximations using quantum arithmetic circuits, but with much lower memory consumption.\footnote{With low-rank approximations, loading classical data into QRAM would only take $\tilde{\mathcal{O}}(k (m+n))$ runtime.} We note that the low-rank approximation assumption of the value matrix has not yet been utilized to produce reductions in resource consumption in classical methods as the PR dynamics and other methods to compute market equilibrium~\citep{gao2020first} require all $mn$ entries of the full value matrix as input.

Storing the results of the computed bids $\hat b^{(t)} \in \mathcal M_{n\times m}(\mathbb{I}_+)$ in QRAM would require a cost of $\tilde{\mathcal O} (mn)$ which would remove all possibility of potential speedups. The same applies to the allocation matrix $\hat x^{(t)}$. Hence, every time we require the usage of $\hat b^{(t)}$ or $\hat x^{(t)}$, we compute them on-the-fly as follows:
\begin{equation}\label{eqOTF}
\textstyle\hat b_{ij}^{(t)} = \frac{B_i^{t+1}v_{ij}^{t}}{m\prod_{k=0}^{t-1}\tilde p_j^{(k)}\prod_{k=0}^{t-1}\tilde \nu_i^{(k)}}, \quad \hat x_{ij}^{(t)} = \frac{B_i^{t+1}v_{ij}^{t}}{m\prod_{k=0}^{t}\tilde p_j^{(k)}\prod_{k=0}^{t-1}\tilde \nu_i^{(k)}}
\end{equation}

Given quantum query access to the values of $(\Pi_{p}^{(t)})_j:= \prod_{k=0}^{t}\tilde p_j^{(k)}$ and $(\Pi_{\nu}^{(t)})_i:= \prod_{k=0}^{t}\tilde \nu_i^{(k)}$, one can encode the values of $\hat b^{(t+1)}$ and $ \hat x^{(t)}$ into a quantum state in superposition via quantum arithmetic circuits in runtime of $\mathcal O(\poly\log tmn)$. In our algorithm, we compute and store the vectors $\Pi_p^{(t)}$ and $\Pi_\nu^{(t)}$ in QRAM, after which, query access to such values cost $\mathcal O(\poly\log mn)$. Taking the $t$-th power of the budget and value cost $\mathcal O(\poly\log t)$ by using CNOT gates and binary exponentiation~\citep{montgomery1987speeding}.

The remaining steps are to compute the price vector $\tilde p$ and utility vector $\tilde u$ in each iteration. Each entry $\tilde p_j$ is the estimation of the $\ell_1$ norm of $\hat b_{*,j}$ and each entry $\tilde \nu_i$ is the estimation of the inner product between $\hat x_{i,*}$ and $v_{i,*}$, which can both be obtained using QAE. $\Pi_p^{(t)}$ and  $\Pi_\nu^{(t)}$ can then be iteratively updated by multiplying by the values of $\tilde p$ and $\tilde \nu$ each iteration. The full algorithm is shown in \cref{algoMain}, and the algorithmic guarantee is shown in \cref{theoremMain}.

\begin{algorithm}
\caption{Quantum algorithm for faulty proportional response dynamics}
\label{algoMain}
\Indm
\KwIn{Quantum query access to $B$ and $v$, Timestep limit $T$, Price error $\varepsilon_p$, Utility error $\varepsilon_\nu$}
\Indp
\DontPrintSemicolon
$\texttt{maxEGVal}= -\infty$, $b_{ij}^{(0)} = \frac{B_i}{m}$\;
\For{$t = 0$ to $T$}{
    \For{$j = 0$ to $m$}{
       $\tilde p_j^{(t)} = \|\hat b_{*, j}^{(t)}\|_1 (1\pm \varepsilon_p)$ via Q norm est. (\cref{lemmaNorm}) with success prob. $1-\frac{\delta}{2mT}$\;
    }
    Store vector $\Pi_p^{(t)} = \tilde p^{(t)} \odot \Pi_p^{(t-1)}$ into QRAM\;\label{alglineqram1}
    Gain access to $\hat x_{ij}^{(t)}$ via $\Pi_p^{(t)}$ and $\Pi_\nu^{(t-1)}$ in QRAM\;
    \For{$i = 0$ to $n$}{
       $\tilde \nu_i^{(t)} = \langle x_{i,*}^{(t)}, v_{i, *}\rangle(1\pm \varepsilon_\nu)$ via Q inner prod. est. (\cref{lemmaQInnerProduct}) with success prob. $1-\frac{\delta}{2nT}$\;
    }
    Store vector $\Pi_\nu^{(t)} = \tilde \nu^{(t)} \odot \Pi_\nu^{(t-1)}$ into QRAM\;\label{alglineqram2}
    Gain access to $\hat b_{ij}^{(t+1)}$ via $\Pi_p^{(t)}$ and $\Pi_\nu^{(t)}$ in QRAM\;
    Classically compute $\tilde{\Phi}^{(t)} = \sum_{i\in[n]} B_i \log(\tilde\nu_i^{(t)})$\;\label{alglinecompute}
    \If{$\tilde{\Phi}^{(t)} > \texttt{maxEGVal}$}{
        $\texttt{maxEGVal} = \tilde{\Phi}^{(t)},\,\texttt{bestPiP} = \Pi_p^{(t-1)},\,\texttt{bestPiNu} = \Pi_\nu^{(t-1)}$\;
    }
}
\Return \texttt{bestPiP} and \texttt{bestPiNu} in QRAM\;
\end{algorithm}

\begin{theorem}[Quantum algorithm for the faulty proportional response dynamics]\label{theoremMain}
Let $\delta \in (0, 0.5), n, m, T\in \mathbb N$, $\varepsilon_p = \frac{\log m}{8T}$, and $\varepsilon_\nu = \frac{\log m}{6T}$. Given quantum query access to $B$ and $v$, and access to QRAM, with success probability $1-\delta$, \cref{algoMain} produces values stored in QRAM such that query and sample access to the values of $\hat b^{(t^*)}$ can be constructed, where
\begin{equation*}
    \textstyle\Phi(\hat b^{(t^*)}) - \Phi(b^*) \le \frac{2\log m}{T},
    \end{equation*}
with $\tilde {\mathcal O}(T^2\sqrt{mn\max(m,n)}\log \frac{1}{\delta})$ runtime and $\tilde {\mathcal O}(m+n)$ QRAM space. To provide query access to $\hat b^{(t^*)}$, an additional cost of $\mathcal O(\poly\log Tmn)$ is incurred from accessing $\Pi_p^{(t^*-1)}$ and $\Pi_\nu^{(t^*-1)}$ in QRAM. Providing sample access to $\hat b^{(t^*)}$ requires additional cost of $\mathcal O(\sqrt{mn}\poly\log T)$ from QAA.
\end{theorem}

\begin{proof}
Per union bound~\citep{boole1847mathematical}, we find that the total success probability is at least $1-\delta$. The output of \cref{algoMain} of \texttt{bestPiP} and \texttt{bestPiNu} corresponds to the values of $\Pi_p^{(t^*-1)}$ and $\Pi_\nu^{(t^*-1)}$. This gives us the guarantee of convergence shown in \cref{theoremEsticonverge}.

For the runtime analysis, the quantum norm estimation subroutine (\cref{lemmaNorm}) takes $\mathcal{O}(\frac{\sqrt{n}}{\varepsilon_p}\log\frac{mT}{\delta})$ for $mT$ iterations, while quantum inner product (\cref{lemmaQInnerProduct}) estimation takes $\mathcal{O}(\frac{\sqrt{m}}{\varepsilon_\nu}\log\frac{nT}{\delta})$ for $nT$ iterations, resulting in a total runtime of $\tilde {\mathcal O}(T^2\sqrt{mn\max(m,n)}\log \frac{1}{\delta})$. For uses of QRAM, the construction on \cref{alglineqram1,alglineqram2} is a one-time cost of $\tilde{\mathcal O} (n)$ and $\tilde{\mathcal O} (m)$, respectively, with a total runtime of $\tilde{\mathcal O} (T(m+n))$. The classical computation of the EG value in \cref{alglinecompute} costs $\mathcal{O}(Tn)$. We note that the quantum norm and inner product estimation subroutine is the main bottleneck of the algorithm, and hence the total runtime is then $\tilde {\mathcal O}(T^2\sqrt{mn\max(m,n)}\log \frac{1}{\delta})$.

For the memory complexity, for the $t$-th iteration, we require 6 vectors in QRAM: the current iteration $\Pi_p^{(t)}$ and $\Pi_\nu^{(t)}$, the best iteration $\texttt{bestPiP}$ and $\texttt{bestPiNu}$ and the previous iteration $\Pi_p^{(t-1)}$ and $\Pi_\nu^{(t-1)}$, in case we need to update $\texttt{bestPiP}$ and $\texttt{bestPiNu}$. Note that to update the best iteration, we simply reroute the register of the previous iteration to being the best iteration. There is no need to copy data or reconstruct a new QRAM as the data from the previous iteration is no longer needed in the next iteration. Therefore, the memory is $\mathcal{O}(m + n)$ for storing the 6 vectors.

The values \texttt{bestPiP} and \texttt{bestPiNu} can be used to construct query access $\hat b^{(t^*)}$ per \cref{eqOTF}. To prepare quantum sample access to the matrix $\hat b^{(t^*)}$, using quantum query access to $\hat b^{(t^*)}$ from \cref{algoMain} and \cref{eqOTF} in addition to conditioned rotation gates, we can encode the values of $\hat b^{(t^*)}$ onto the amplitudes of an ancilla qubit via \cref{lemmaNorm} such that we obtain
\begin{equation}
\textstyle\frac{1}{\sqrt{mn}}\sum\limits_{i,j}\bm{\ket{i}}\bm{\ket{j}}\ket{0} \to \frac{1}{\sqrt{mn}} \sum\limits_{i,j}\bm{\ket{i}}\bm{\ket{j}}\Big(\sqrt{\hat b^{(t^*)}_{ij}}\ket{0} + \sqrt{1-\hat b^{(t^*)}_{ij}}\ket{1}\Big)
\end{equation}

Note that states where the last qubit is $\ket{0}$ are the ``good'' states that we want, while the states where the last qubit is $\ket{1}$ are ``bad'' states. To get rid of the ``bad'' states, we use QAA to amplify the amplitudes of the ``good'' states. Rewriting the above terms and ignoring the ancilla qubits, we obtain 
\begin{equation}
\textstyle\sqrt{\frac{\|\hat b^{(t^*)}\|_1}{mn}}\sum\limits_{i,j}\bm{\ket{i}}\bm{\ket{j}}\sqrt{\frac{\hat b^{(t^*)}_{ij}}{\|\hat b^{(t^*)}\|_1}}\ket{0} + \sqrt{1-\frac{\|\hat b^{(t^*)}\|_1}{mn}}\sum\limits_{i,j}\bm{\ket{i}}\bm{\ket{j}}\sqrt{\frac{1-\hat b^{(t^*)}_{ij}}{mn-\|\hat b^{(t^*)}\|_1}}\ket{1}
\end{equation}
We see from the above that the success probability of obtaining ``good'' states is $\frac{\|\hat b^{(t^*)}\|_1}{mn}$. Hence, QAA costs $\mathcal{O}\Big(\sqrt{\frac{mn}{\|\hat b^{(t^*)}\|_1}}\Big)$ query accesses to $b^{(t^*)}$, and while the exact value of $\|\hat b^{(t^*)}\|_1$ is not known, by \cref{defFPRupdate}, $\frac{1}{1+\varepsilon_\nu} \le \|\hat b^{(t^*)}\|_1 \le \frac{1}{1-\varepsilon_\nu}$. Hence, the query complexity of $b^{(t^*)}$ can be further upper bounded as $\mathcal{O}(\sqrt{mn(1+\varepsilon_\nu)})\subset\mathcal{O}(\sqrt{mn})$. On the other hand, the cost of query access to $b^{(t^*)}$ has $\mathcal O(\poly\log Tmn)$ cost, therefore, the total cost for sample access to $b^{(t^*)}$ is $\tilde{\mathcal O}(\sqrt{mn}\poly\log T)$.

\end{proof}

\section{Numerical simulations}
We simulate the market equilibrium computation under PR dynamics and our quantum algorithm.\footnote{The codebase can be found at \url{https://github.com/georgepwhuang/q-market-equilibrium}.} To showcase the effects of quantum speedups, we fixed the number of queries to all bid matrices $b^{(t)}$ and observed the reduction of the objective value over the number of queries.

As an actual simulation of QAE using quantum gates over multiple qubits is costly, we directly compute the probability vector of $\Pr[Z=z]$ for $z\in [M]$ that one would obtain by QAE~\citep{brassard2002quantum} for a target value $a$,
\begin{equation}
\textstyle\Pr[Z = z] = \frac{\sin^2(M\Delta_z \pi)}{M^2\sin^2(\Delta_z \pi)}
\end{equation}
where $\Delta_z = \min(\lvert z - \sin^{-1}(\sqrt{a})/\pi\rvert, \lvert 1- z + \sin^{-1}(\sqrt{a})/\pi\rvert)$. $M$ is the number of times that call the unitaries $U$ and $V$ in QAE (see \cref{theoremQAE}), and is linearly correlated to the runtime. We then sample the output according to the computed probabilities to obtain an estimator $\tilde a = \sin^2(\pi \frac{z}{M})$. 

For our experiments, we generate the input data $v$ where the value $v$ is sampled from a uniform distribution with range $(0, 1]$ and a normal distribution $\mathcal{N} (0.5, 0.25)$, where we resample values that fall outside the range of $(0, 1]$. For the budget $B$, we either sample from the same distribution as the value matrix or set the same budget for all buyers to simulate competitive equilibrium from equal income (CEEI) applications. Our simulation includes $n=16384$ buyers, $m=16384$ goods, and iterate for $T=16$ iterations for the PR dynamics. For the quantum algorithm, note that the queries per iteration would be reduced by $\sqrt{n}$ if we use an actual quantum computer, hence increasing the number of iterations to fix the number of queries. For QAE, we run for $\sqrt{T\sqrt{n}} = 512$ iterations and set $M \in \mathcal{O}(\sqrt{T\sqrt{n}})$. As the classical algorithms are deterministic, we rerun our quantum algorithm over $15$ times with the same sample of $B$ and $v$ to observe the variance of convergence progress. Results are shown in \cref{figExperiments}. Details on implementation and setup are found in \cref{appendixExp}.

\begin{figure}
    \centering
    \begin{subfigure}{\textwidth}
    \caption{}
    \includegraphics[width=\linewidth]{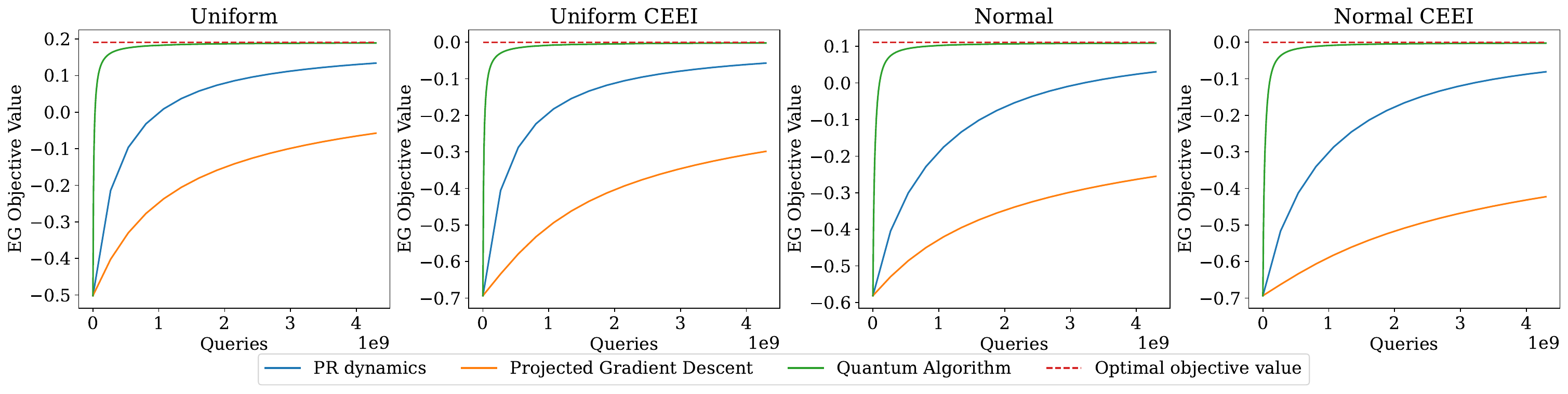}
    \label{figNumerics}
    \end{subfigure}
    \begin{subfigure}{\textwidth}
    \caption{}
    \includegraphics[width=\linewidth]{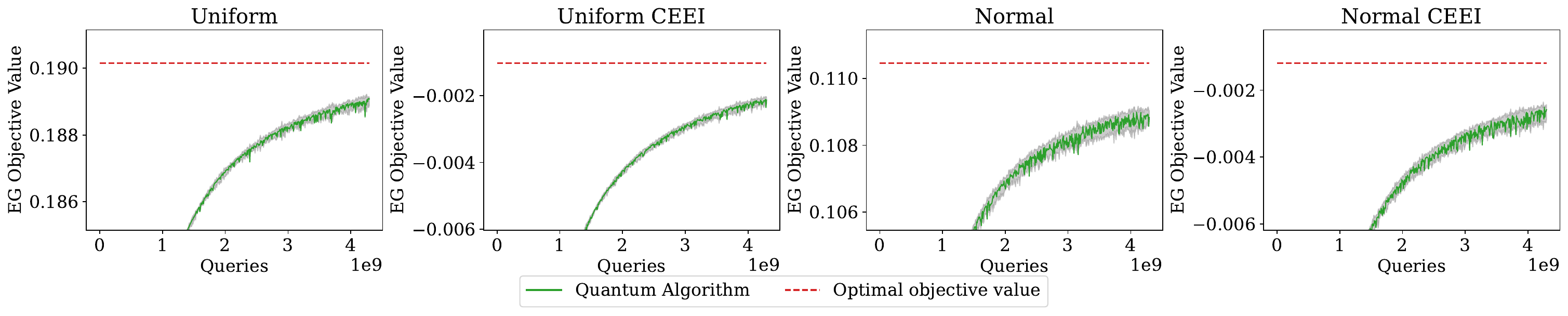}
    \label{figZoom}
    \end{subfigure}
    \caption{\textit{Experimental results.} We perform our experiments on $n=16384$ buyers and $m=16384$ goods given the same amount of queries for all algorithms. We observe in \cref{figNumerics} that over different distributions, our quantum algorithm (green) significantly outperforms the PR dynamics (blue), which aligns with our theoretical results. Furthermore, our results also show that both our quantum algorithm and the PR dynamics outperform projected gradient descent (orange) in the mid-accuracy regime. \cref{figZoom} shows the convergence of a single run of the quantum algorithm despite its instability from faulty updates, as well as the variance over the multiple runs (shaded in gray).}
    \label{figExperiments}
\end{figure}

From the plots of \cref{figExperiments}, we note that the results fit our theoretical results in that the quantum algorithm converges much faster than that of the PR dynamics~\citep{wu2007proportional}. Further, we also compare against the convergence of projected gradient descent, which supports empirical results by \citet{gao2020first} that in the regime of mid-level accuracy and low iterations, PR dynamics-related algorithms, both classical and quantum, converge faster than projected gradient descent.

\section{Discussion}\label{sectDiscussion}
\paragraph{Quasi-linear utilities.}
For the bulk of our paper, we focus on the setting of linear utilities for Fisher markets. However, applications of market equilibrium computation in large-scale Fisher markets involve mostly quasi-linear utilities~\citep{kroer2021market}.
An approach for using PR dynamics for quasi-linear utilities proposed by~\citet{gao2020first}\footnote{There is another method proposed by \citet{cheung2021learning}, which we find difficult to convert to quantum due to its use of thresholding, which would cause problems with faulty updates from the FPR dynamics.} includes the usage of slack variables $\delta = (\delta_1, \cdots, \delta_m)$ that represent the buyers' leftover budgets. The PR updates are then modified as follows:
\begin{equation}\label{eqQLPRupdate}
    \textstyle b_{ij}^{(t+1)} = B_i \frac{v_{ij}x_{ij}^{(t)}}{\sum_{j'}v_{ij'}x_{ij'}^{(t)}+ \delta_i^{(t)}},\; \delta_i^{(t+1)} = B_i \frac{\delta_i^{(t)}}{\sum_{j'}v_{ij'}x_{ij'}^{(t)}+ \delta_i^{(t)}}.
\end{equation}

Further, PR dynamics for quasi-linear utilities exhibit a convergence rate of $\mathcal{O}\big(\frac{\log (m+1)}{T}\big)$. Using the methods discussed in previous sections, the quasi-linear version of PR dynamics can then be readily adapted to its quantum version by employing the same techniques of computing and storing in QRAM the values $\Pi_p^{(t)}$ and $\Pi_\nu^{(t)}$ with on-the-fly computation of $\hat b^{(t+1)}$, $\hat x^{(t)}$ and $\hat \delta^{(t+1)}$.

\paragraph{Constant number of buyers.}
Notice that our quantum algorithm provides a quadratic speedup on the smaller value in regards to the number of buyers $n$ and number of goods $m$. Therefore given extreme cases where the number of buyers $n \in \mathcal{O}(1)$, our algorithm does not provide a speedup. However, in such cases, quantum speedups may still be obtained simply by removing the QAE step for estimating the price for each item and replacing it with using quantum arithmetic circuits to compute the exact sum. We use a total of $\mathcal O(nT\poly\log(T,m,n))$ qubits to compute the values of $b_{i,*}$ separately on the $i$-th set of qubits, and only conduct QAE when estimating the utility value for each buyer. Given that in this setting, $n \in \mathcal{O}(1)$, the total runtime would then be $\mathcal{O}\big(T^2 \sqrt{m}\log \frac{1}{\delta}\big)$, gaining a quadratic speedup over the number of goods $m$. 

\paragraph{Dequantization.}
Given the work in recent years towards the development of quantum-inspired classical algorithms~\citep{tang2019quantum, arrazola2020quantum, tang2021quantum, chia2022sampling} that achieve similar performances as quantum algorithms using sampling-based techniques, a natural question that arises is whether our algorithm can be ``de-quantized''. The main speedup in our algorithm stems from the usage of estimation of $\ell_1$ norms and inner products. While the use of sampling techniques can indeed provide inner product estimations, they retain the same $\mathcal{O} (1/\varepsilon^2)$ dependency instead of the $\mathcal{O} (1/\varepsilon)$ dependency of QAE. Hence, our algorithm performance may be hard to replicate in classical settings.

On the other hand, while it has been suggested that the computation of market equilibrium may benefit from low-rank approximations~\citep{kroer2022computing}, methods of using such properties to accelerate the computation of gradients have not been proposed, given that the update of the PR dynamics rely on element-wise multiplication of matrices instead matrix multiplication. This would suggest that using sampling techniques to accelerate updates would be similarly difficult.

\paragraph{Practicality.}
Our quantum algorithm relies mostly on the QAE subroutine to achieve quantum speedups. In its original formulation, QAE involves the usage of QPE as a subroutine, which requires multiple controlled unitaries and the quantum Fourier transform~\citep{coppersmith1994approximate}. QAE is thus regarded as a fault-tolerant quantum subroutine, whose execution may be beyond the capabilities of current quantum hardware. However, many subtle improvements to the QAE algorithm have since been made after its discovery, such as simplifying subroutines~\citep{suzuki2020amplitude, grinko2021iterative, nakaji2020faster, aaronson2020quantum, rall2023amplitude,labib2024quantum} or trading circuit depth with speedup factors~\citep{rall2023amplitude, giurgicatiron2022low,vu2024low}. Such improvements pave the way for potential implementation of QAE, and by extension, this paper, on next-generation quantum hardware in the early fault-tolerant regime~\citep{campbell2021early, katabarwa2024early}.

\paragraph{Potential and limitations for further quantum speedups.} Our quantum algorithm shares similarities to other quantum algorithms that are based on the \emph{multiplicative weight update} (MWU) method~\citep{arora2005fast, arora2012multiplicative}. Such methods have found success in obtaining quantum speedups for LPs~\citep{vanapeldoorn2019quantum} and SDPs~\citep{brandao2017quantum,brandao2019quantum,vanapeldoorn2019improvements,vanapeldoorn2020quantum}, which have been extended to applications such as zero-sum games~\citep{vanapeldoorn2019quantum, jain2022matrix}, quadratic binary optimization~\citep{brandao2022faster}, and financial applications~\citep{rebentrost2021quantum,lim2024quantum}. Apart from the MWU-esque PR dynamics, various other methods for computing market equilibrium have also been proposed. Can quantum speedups obtained from these methods exceed those of our quantum algorithm? 

Tracing back to the roots of the EG convex program~\citep{eisenberg1959consensus, eisenberg1961aggregation} and Shmyrev convex program~\citep{shmyrev2009algorithm}, it is well known that such programs can be solved in polynomial time with interior-point methods (IPM)~\citep{boyd2004interior}. However, as IPMs require using linear solvers as subroutines, and as there is no guarantee of well-conditioned systems, the quantum linear systems solver~\citep{harrow2009quantum, childs2017quantum} may not provide significant speedup. Therefore, it may be unlikely that quantum IPMs~\citep{kerenidis2020quantum} can provide significant speedups.

First-order methods such as the Frank-Wolfe (FW) algorithm~\citep{frank1956algorithm} and projected gradient descent (PGD) have also been discussed as candidates for solving market equilibrium~\citep{gao2020first}, with PGD achieving linear convergence classically. While PGD obtains a superior asymptotic convergence rate in terms of the error $\varepsilon$ compared to PR dynamics, as our quantum speedups stem from faster computations of results within a single iteration, it may be harder to find such speedups for PGD as there has been no evidence for quantum speedups in projections onto a simplex~\citep{duchi2008efficient, condat2015fast} as required. 

On the other hand, the FW algorithm has been shown to provide quantum speedups for regression~\citep{du2022quantum, chen2023quantum}. However, convergence results of FW \citep{clarkson2010coresets, jaggi2013revisiting} show that $\Phi(b^{(T)}) - \Phi(b^*) \le \frac{C_\Phi}{(T+2)}$, where $C_\Phi$ can be shown to be $\mathcal{O}(n)$ by computing relevant values.
The number of iterations $T$ required for convergence to additive error $\varepsilon$ is then $\mathcal{O}(\frac{n}{\varepsilon})$ as compared to $\mathcal{O}(\frac{\log m}{\varepsilon})$ of PR dynamics. This matches the results of \citet{gao2020first}, which show that FW has slow convergence empirically for market equilibrium computation. Prior no-go results suggest that quantum algorithms cannot provide speedups for $T$ when $T$ is independent of the problem dimension~\citep{garg2021no, garg2021optimal}. Assuming no quantum speedups in $T$, given the $\mathcal{O}(n)$ upper bound in the FW algorithm, the quantum algorithm based on FW can potentially have a higher dependency on $n$ than the classical PR dynamics. 

\begin{ack}
The authors thank Gregory Kang Ruey Lau, Jinge Bao, and Warut Suksompong for discussions. This work is supported by the National Research Foundation, Singapore, and A*STAR under its CQT Bridging Grant and its Quantum Engineering Programme under grant NRF2021-QEP2-02-P05.
\end{ack}

\nocite{apsrev42Control}
{\small 
\bibliographystyle{apsrev4-2} 
\bibliography{main}}

\newpage
\appendix

\counterwithin{theorem}{section}
\counterwithin{proposition}{section}
\counterwithin{lemma}{section}
\counterwithin{equation}{section}

\section{Quantum subroutines}
\label{appendixQroutine}
In this section, we show prior results that obtain $\ell_1$ norms and inner products with quadratic speedups with QAE.

\begin{lemma}[Quantum state preparation and norm estimation; Lemma 5,~\citep{rebentrost2021quantum}]\label{lemmaNorm}
Let $n \in \mathbb N$. We are given quantum query access to non-zero vector $w \in \mathbb{I}^n$, with $\max_j w_j = 1$.
\begin{enumerate}[leftmargin=*]
\item There exists a quantum circuit that prepares the state $\frac{1}{\sqrt{n}}  \sum_{j=1}^n \bm{\ket j}  \left( \sqrt{ w_j } \ket{0} + \sqrt{1- w_j} \ket{1} \right)$ with two queries and $\mathcal O \left(\log n\right)$ gates.
\item Let $\varepsilon>0$ and $\delta \in (0, 1)$. There exists a quantum algorithm that provides an estimate $\Gamma_w$ of the $\ell_1$-norm $\|w\|_1$ such that $\left| \|w\|_1 - \Gamma_w\right| \le \varepsilon \|w\|_1$, with probability at least $1-\delta$. The algorithm requires $\mathcal O\left(\frac{ \sqrt{n} }{\varepsilon} \log \frac{1}{\delta} \right)$ queries and $\tilde{\mathcal O}\left(\frac{ \sqrt{n} }{\varepsilon} \log \frac{1}{\delta} \right)$quantum gates.
 \end{enumerate}
\end{lemma}
\begin{proof}
We reiterate the proof of Lemma 5 in \citep{rebentrost2021quantum} for the convenience of the reader.
\begin{enumerate}[leftmargin=*]
\item First, using $\mathcal{O}(\log n)$ Hadamard gates, prepare the state $\frac{1}{\sqrt{n}}\sum_{j=1}^n \bm{\ket j}$. Then, by quantum query access to $w$, obtain $\frac{1}{\sqrt{n}}\sum_{j=1}^n \bm{\ket j}\ket{w_j}$. By controlled rotation gates, we can then obtain $\frac{1}{\sqrt{n}}\sum_{j=1}^n \bm{\ket j}\ket{w_j}(\sqrt{w_j}\ket{0}+\sqrt{1-w_j}\ket{1})$. By another quantum query access to $w$, we can uncompute the intermediate registers and obtain $\frac{1}{\sqrt{n}}\sum_{j=1}^n \bm{\ket j}(\sqrt{w_j}\ket{0}+\sqrt{1-w_j}\ket{1})$.
\item First observe that that with projector $P = I_n \otimes \ketbra{0}{0}$ and $\ket{\psi} = \frac{1}{\sqrt{n}}\sum_{j=1}^n \bm{\ket j}(\sqrt{w_j}\ket{0}+\sqrt{1-w_j}\ket{1})$, one can obtain $a = \|P\ket{\psi}\|_2^2 = \frac{\|w\|_1}{n}$. Setting $M\ge \frac{6\pi}{\varepsilon}\sqrt{n}$, we obtain an estimate
\begin{equation}
    |\tilde a_{\mathrm{est}} - a| \le  2\pi\frac{\sqrt{a(1-a)}}{M}+\frac{\pi^2}{M^2} \le \frac{\varepsilon}{6\sqrt{n}}\left(2\sqrt{a}+\frac{\varepsilon}{12}\right) \le \frac{3\sqrt{a}\varepsilon}{6\sqrt{n}} \le  \frac{\sqrt{\|w\|_1}\cdot\varepsilon}{2n} \le \frac{a}{2}\cdot\varepsilon
\end{equation}
with probability at least $\frac{8}{\pi^2}$. Using the powering lemma~\citep{jerrum1986random}, we can boost the success probability to $1-\delta$ by taking the median of $\mathcal O(\log \frac{1}{\delta})$ runs of the QAE algorithm.
\end{enumerate}
\end{proof}
\begin{remark}
Note that~\cref{lemmaNorm} has the requirement that $\max_j w_j =1$. For cases where this is not the case, we can use a maximum finding algorithm to divide all entries by the largest value. Such can be achieved by the following quantum minimum/maximum finding algorithm in $\mathcal{O}(\sqrt{n})$ runtime, which we introduce below. Recall that division takes $\mathcal{O}(1)$ runtime with quantum arithmetic circuits.
\end{remark}
\begin{lemma}[Quantum minimum finding; Theorem 1,~\citep{durr1996quantum}]
\label{lemmaQMin}
Let $n \in \mathbb N$. Given quantum query access to non-zero vector $w \in \mathbb{I}^n$, we can find the minimum $w_{\min} = \min_{j\in[n]}  w_j$ with success probability $1-\delta$ with $\mathcal O\left(\sqrt{n} \log \frac{1}{\delta} \right)$ queries and $\tilde{\mathcal O}\left(\sqrt{n} \log \frac{1}{\delta} \right) $quantum gates.
\end{lemma}
\begin{corollary}[Quantum maximum finding]\label{corollaryQMax}
    Let $n \in \mathbb N$. Given quantum query access to non-zero vector $w \in \mathbb{I}^n$, we can find the maximum $w_{\max} = \max_{j\in[n]}  w_j$ with success probability $1-\delta$ with $\mathcal O\left(\sqrt{n} \log \frac{1}{\delta} \right)$ queries and $\tilde{\mathcal O}\left(\sqrt{n} \log \frac{1}{\delta} \right) $quantum gates.
\end{corollary}

Below we present a quantum inner product estimation algorithm simplified from Lemma 6 of \citep{rebentrost2021quantum}.
\begin{lemma} [Quantum inner product estimation with relative accuracy] \label{lemmaQInnerProduct}
Let $n \in \mathbb N$, $\varepsilon < 0$ and $\delta \in(0,1)$. We are given quantum query access to two vectors $u,v \in \mathbb{I}^n$.
An estimate $\Gamma_{u,v}$ for the inner product can be provided such that $|\Gamma_{u,v} - u\cdot v | \le \varepsilon\  u\cdot v$ with success probability $1-\delta$.
This estimate is obtained with $\mathcal O\left(\frac{ \sqrt{n} }{\varepsilon} \log \frac{1}{\delta} \right)$ queries and $\tilde{\mathcal O}\left(\frac{ \sqrt{n} }{\varepsilon} \log \frac{1}{\delta} \right) $quantum gates.
\end{lemma}
\begin{proof}
Using quantum arithmetic circuits, we can obtain $z_j = u_jv_j$, i.e., $z = u \odot v$, by the following:
\begin{equation}
\bm{\ket{j}} \to \bm{\ket{j}}\ket{u_j}\ket{v_j} \to \bm{\ket{j}}\ket{z_j}\ket{v_j} \to \bm{\ket{j}}\ket{z_j}\ket{\bar 0}
\end{equation}
 Using quantum maximum finding in \cref{corollaryQMax} to find $z_{\max}$ up to probability $1-\delta/2$, we can then obtain $z_j/z_{\max}$. Lastly, using \cref{lemmaNorm}, we can obtain $\Gamma_{u,v}  = \Gamma_z$ such that $|\Gamma_{u,v} -u\cdot v| = |\Gamma_z - \|z\|_1| \le \varepsilon u\cdot v$ up to probability $1-\delta/2$. Using a union bound~\citep{boole1847mathematical}, we find the total success probability of the entire process is $1-\delta$.
\end{proof}

\section{Convergence guarantees for the PR dynamics}
\label{appendixPRProof}

We show the convergence guarantee of the proportional response (PR) dynamics in regards to the Eisenberg-Gale convex program by~\citet{zhang2011proportional} and improved upon by~\citet{birnbaum2011distributed}, and the convergence in regards to the Shmyrev convex program, first shown also by~\citet{birnbaum2011distributed}.  Recall that the negative target function from the Eisenberg-Gale convex program is 
\begin{equation}
\Phi(b) = -\sum_{i \in [n]} B_i\log u_i,
\end{equation}
and the negative target function from the Shmyrev convex program is 
\begin{equation}
\Psi(b) = -\sum_{i \in [n], j\in[m]} b_{ij}\log\frac{v_{ij}}{p_j}.
\end{equation} 

We first set up the following convex set:
\begin{equation}
\mathcal{B} = \left\{b \in \mathcal M_{n\times m}(\mathbb R_+): \sum_j b_{ij} = B_i\right\}
\end{equation}

To show the convergence of the PR dynamics, we first need the following inequalities:
\begin{lemma}\label{lemmaShmEqEG}
Let $b^* = \argmin_{b\in \mathcal{B}} \Phi(b)$. Then $\Phi(b^*) = \Psi(b^*) - \sum_{i \in [n]}B_i\log B_i$.
\end{lemma}

\begin{proof}
Recall that the price $p_j$ is the dual variable of the constraint on $x_{ij}$ for the Eisenberg-Gale convex program, and that $u_i = \sum_j x_{ij}v_{ij}$. By the KKT stationarity constraint, we see that 
\begin{equation}
 \frac{\partial}{\partial x_{ij}} \left(-\sum_iB_i \log u_i + \sum_j p_j\left(\sum_i x_{ij}-1\right)\right) = -\frac{B_iv_{ij}}{u_i} + p_j = 0
\end{equation}
from which we can infer that $\frac{B_iv_{ij}}{u_i^*} = p_j^*$
and show $\Phi(b^*) = \Psi(b^*)-\sum_{i \in [n]}B_i\log B_i$. We have that,
\begin{align}
\Phi(b^*) &= -\sum_{i \in [n]} B_i\log u_i^* = -\sum_{i \in [n], j\in[m]} b_{ij}^*\log u_i^* = -\sum_{i \in [n], j\in[m]} b_{ij}^*\log B_i \frac{v_{ij}}{p_j^*} \\ =&
-\sum_{i \in [n], j \in [m]} b_{ij}^* \log \frac{v_{ij}}{p_j^*} - \sum_{i \in [n]}B_i\log B_i
= \Psi(b^*) - \sum_{i \in [n]}B_i\log B_i.
\end{align}
\end{proof}

\begin{lemma}[Lemma 19, \citep{birnbaum2011distributed}]\label{lemmaEGLTShm}
$\forall b \in \mathcal{B}, \Phi(b) \le \Psi(b)-\sum_{i \in [n]} B_i \log B_i.$
\end{lemma}
\begin{proof}
We reiterate the proof of Lemma 19 in \citep{birnbaum2011distributed} for the convenience of the reader. By convexity of $-\log$, we see that,
\begin{align}
\Phi(b) &= -\sum_{i \in [n]} B_i\log u_i = -\sum_{i \in [n]} B_i\log \sum_{j \in [m]} \frac{b_{ij}}{p_j}v_{ij}\\
&= -\sum_{i \in [n]}B_i\log \sum_{j\in [m]} \frac{b_{ij}}{B_i}\frac{v_{ij}}{p_j} - \sum_{i \in [n]}B_i\log B_i\\
&\le -\sum_{i \in [n], j \in [m]}\frac{b_{ij}}{B_i}B_i\log \frac{v_{ij}}{p_j} - \sum_{i \in [n]}B_i\log B_i\\
&= \Psi(b) - \sum_{i \in [n]}B_i\log B_i.
\end{align}
\end{proof}

\begin{lemma}\label{lemmaShmLTEG}
Let $\displaystyle b_{ij}' = B_i\frac{b_{ij}v_{ij}/p_j}{u_i}$. Then $\forall b \in \mathcal{B}, \Psi(b') \le \Phi(b) + \sum_{i \in [n]} B_i \log B_i.$
\end{lemma}
\begin{proof}
 Let $p_j' = \sum_i b_{ij}'$. By concavity of $\log$:
\begin{align}
\Psi(b') &= \sum_{i \in [n], j \in [m]} b_{ij}' \log \frac{p_j'}{v_{ij}} = \sum_{i \in [n], j \in [m]} B_i\frac{b_{ij}v_{ij}/p_j}{u_i} \log \frac{p_j'}{v_{ij}}\\
&\le \sum_{i \in [n]} B_i \log \sum_{j\in[m]} \frac{b_{ij}v_{ij}/p_j}{u_i}\frac{p_j'}{v_{ij}} = \sum_{i \in [n]} B_i \log \left(\frac{1}{u_i}\sum_{j\in[m]} \frac{b_{ij}p_j'}{p_j}\right)\\
&= \sum_{i \in [n]} B_i \log \sum_{j\in[m]} \frac{b_{ij}p_j'}{B_ip_j} - \sum_{i \in [n]} B_i \log u_i + \sum_{i \in [n]} B_i \log B_i\\
&\le \log \sum_{i\in[n], j\in[m]} \frac{b_{ij}p_j'}{p_j} - \sum_{i \in [n]} B_i \log u_i + \sum_{i \in [n]} B_i \log B_i\\
&= \log \sum_{j\in[m]} p_j' - \sum_{i \in [n]} B_i \log u_i + \sum_{i \in [n]} B_i \log B_i\\
&= \Phi(b) + \sum_{i \in [n]} B_i \log B_i.
\end{align}
\end{proof}

From the above two lemmas, we gain the monotonically decreasing properties of iteratively updating $b$ via the PR dynamics on the negative target functions of the Eisenberg-Gale and Shmyrev convex programs:

\begin{lemma}\label{lemmaEGMonoDec}
$\forall t\ge0, \Phi(b^{(t+1)}) \le \Phi(b^{(t)})$.
\end{lemma}
\begin{proof}
Apply \cref{lemmaEGLTShm} and \cref{lemmaShmLTEG} consequently.
\end{proof}

\begin{corollary}
[Lemma 5, \citep{birnbaum2011distributed}]\label{corollaryShmMonoDec}
$\forall t\ge0, \Psi(b^{(t+1)}) \le \Psi(b^{(t)})$.
\end{corollary}

We now use the following lemmas to construct an end-to-end proof of the convergence of the PR dynamics. In a slight abuse of notation, we adapt the definition of KL divergence to matrices such that for $u, v \in \mathcal M(\mathbb R_+)_{p\times q}$, let $D(u\|v):= \sum_{i\in[p], j\in[q]} u_{ij}\log\frac{u_{ij}}{v_{ij}}$. The following can then be shown:

\begin{lemma}\label{lemmaEGtelescope}
$\forall t\ge0, \sum_{t=0}^T \left(\Phi(b^{(t)}) - \Phi(b^*)\right) \le D(b^*\|b^{(0)})$.
\end{lemma}
\begin{proof}
Similar by the proof of Theorem 3 of \citep{zhang2011proportional}, we first lower bound $\Delta_t = D(b^*\|b^{(t)}) - D(b^*\|b^{(t+1)})$ as follows:
\begin{align}
    \Delta_t &= D(b^*\|b^{(t)}) - D(b^*\|b^{(t+1)})\\
    &= \sum_{i\in[n], j \in [m]} b_{ij}^*\log \frac{b^{(t+1)}_{ij}}{b^{(t)}_{ij}}= \sum_{i\in[n], j \in [m]} b_{ij}^*\log \frac{B_iv_{ij}}{p_j^{(t)}u_i^{(t)}}\\
    &= \sum_{i\in[n], j \in [m]} \left( b_{ij}^*\log \frac{v_{ij}}{p_j^*} +  b_{ij}^*\log \frac{p_j^*}{p_j^{(t)}} - b_{ij}^*\log u_i^{(t)} +  b_{ij}^*\log B_i\right)\\
    &= \sum_{i\in[n], j \in [m]} b_{ij}^*\log \frac{v_{ij}}{p_j^*} + \sum_{j \in [m]} p_j^*\log \frac{p_j^*}{p_j^{(t)}} - \sum_{i\in[n]} B_i\log u_i^{(t)} + \sum_{i\in[n]} B_i\log B_i\\ 
    &= -\Psi(b^*) + D(p_j^*\|p_j^{(t)}) + \Phi(b^{(t)}) + \sum_{i\in[n]} B_i\log B_i\\
    &=  D(p_j^*\|p_j^{(t)}) + \Phi(b^{(t)}) - \Phi(b^*)\\
    &\ge \Phi(b^{(t)}) - \Phi(b^*)
\end{align}
where the second-to-last equality is by \cref{lemmaShmEqEG} and the inequality is by the positivity of KL divergence.
Taking the telescoping sum of $\Delta_t$, we see that
\begin{equation}
\sum_{t=0}^T \Delta_t = \sum_{t=0}^T D(b^*\|b^{(t)}) - D(b^*\|b^{(t+1)}) = D(b^*\|b^{(0)}) - D(b^*\|b^{(T+1)}) \ge \sum_{t=0}^T (\Phi(b^{(t)}) - \Phi(b^*))
\end{equation}
Hence, we obtain $\sum_{t=0}^T \left(\Phi(b^{(t)}) - \Phi(b^*)\right) \le D(b^*\|b^{(0)})$. 
\end{proof}

\begin{proposition}\label{propositionEGConverge}
$\forall t\ge0, \Phi(b^{(T-1)})-\Phi(b^*) \le \frac{D(b^{(0)}\|b^*)}{T}$.
\end{proposition}
\begin{proof}
Combining \cref{lemmaEGMonoDec} and \cref{lemmaEGtelescope}, we can write
\begin{equation}
\Phi(b^{(T-1)}) - \Phi(b^*) \le \frac{1}{T} \sum_{t=0}^{T-1}\Phi(b^{(t)}) - \Phi(b^*) \le \frac{D(b^*\|b^{(0)})}{T}.
\end{equation}
\end{proof}

\begin{corollary}[Lemma 3, \citep{birnbaum2011distributed}]\label{corollaryShmConverge}
    $\forall t\ge0, \Psi(b^{(T)})-\Psi(b^*) \le \frac{D(b^{(0)}\|b^*)}{T}$.
\end{corollary}
\begin{proof}
Apply \cref{lemmaShmLTEG} to \cref{propositionEGConverge}.
\end{proof}

Lastly, we can upper bound the value $D(b^*\|b^{(0)})$ in terms of dimensions $m$ and $n$ given that each buyer initially divides the budget equally between all items such that $b^{(0)}_{ij} = \frac{B_i}{m}$.

\begin{lemma}[Lemma 13, \citep{birnbaum2011distributed}; Theorem 7, \citep{gao2020first}]\label{lemmaInit}
If $b_{ij}^{(0)} = \frac{B_i}{m}$ for all $i$ and $j$, then $D(b^*\|b^{(0)}) \le \log m$.
\end{lemma}
\begin{proof}
Evaluating $D(b^*\|b^{(0)})$, we have
\begin{equation}
D(b^*\|b^{(0)}) = \sum_{ij} b_{ij}^* \log\frac{b_{ij}^*}{b_{ij}^{(0)}} = \sum_{ij} b_{ij}^* \log\frac{mb_{ij}^*}{B_i}=\log m + \sum_{ij} b_{ij}^* \log \frac{b_{ij}^*}{B_i}\le \log m,
\end{equation}
as $b_{ij}^* \le B$, and $\log \frac{b_{ij}^*}{B_i}$ is negative. 
\end{proof}

Plugging \cref{lemmaInit} into \cref{propositionEGConverge} and \cref{corollaryShmConverge}, we obtain the convergence guarantee of \cref{theoremPRconverge}.

\section{Convergence guarantees for the FPR dynamics}\label{appendixFPRProof}

In this section, we prove the convergence guarantee of the faulty proportional response (FPR) dynamics. We first examine the immediate effects of allowing erroneous estimations of $u$ and $p$ in the FPR dynamics. Let $\hat B_i^{(t)} = \sum_{j\in[m]} \hat b_{ij}^{(t)}$. Note that $\hat B_i^{(t)}\ne B_i$ as the normalization step of constructing $\hat b_{ij}$ is erroneous. By the construction of $\hat b^{(t)}$ by the FPR dynamics, 
\begin{equation}
\hat b_{ij}^{(t)} = \frac{B_i}{\tilde \nu_i^{(t-1)}} v_{ij}\frac{\hat b_{ij}^{(t-1)}}{\tilde p_{j}^{(t-1)}},
\end{equation}
we can find that
\begin{equation}
\hat B_i^{(t)} =  \sum_{j\in[m]} \hat b_{ij}^{(t)} =  \sum_{j\in[m]}\frac{B_i}{\tilde \nu_i^{(t-1)}} v_{ij}\frac{\hat b_{ij}^{(t-1)}}{\tilde p_{j}^{(t-1)}} = \frac{B_i}{\tilde \nu_i^{(t-1)}} \sum_{j\in[m]}v_{ij}\frac{\hat b_{ij}^{(t-1)}}{\tilde p_{j}^{(t-1)}} = B_i\frac{\hat \nu_i^{(t-1)}}{\tilde \nu_i^{(t-1)}},
\end{equation}
where we can obtain the following inequality by definition of $\tilde \nu_i$:
\begin{equation}
\frac{B_i}{1+\varepsilon_\nu} \le \hat B_i^{(t)} \le \frac{B_i}{1-\varepsilon_\nu}
\end{equation}
By summing $\hat{B_i}$, we find that 
\begin{equation}
\frac{1}{1+\varepsilon_\nu} \le \sum_{i\in [n]} \hat B_i^{(t)} = \sum_{i\in [n], j \in [m]} \hat b_{ij}^{(t)} = \sum_{j \in [m]} \hat p_{j}^{(t)}  \le \frac{1}{1-\varepsilon_\nu}.
\end{equation}

We now prove \cref{theoremFPRconverge}.

\fprconverge*

\begin{proof}
Similar to the proof of \cref{lemmaEGtelescope}, we first lower bound $\hat\Delta_t = \sum_{i\in[n], j \in [m]} b_{ij}^*\log \frac{b_{ij}^*}{\hat b^{(t)}_{ij}}-\sum_{i\in[n], j \in [m]} b_{ij}^*\log \frac{b_{ij}^*}{\hat b^{(t+1)}_{ij}}$, where we use as follows:
\begin{align}
\hat\Delta_t 
    &= \sum_{i\in[n], j \in [m]} b_{ij}^*\log \frac{\hat b^{(t+1)}_{ij}}{\hat b^{(t)}_{ij}}= \sum_{i\in[n], j \in [m]} b_{ij}^*\log \frac{B_iv_{ij}}{\tilde p_j^{(t)}\tilde \nu_i^{(t)}}\\
    &= \sum_{i\in[n], j \in [m]} \left(b_{ij}^*\log \frac{v_{ij}}{p_j^*} + b_{ij}^*\log \frac{p_j^*}{\tilde p_j^{(t)}} - b_{ij}^*\log \tilde \nu_i^{(t)} + b_{ij}^*\log B_i\right)\\
    &= \sum_{i\in[n], j \in [m]} b_{ij}^*\log \frac{v_{ij}}{p_j^*} + \sum_{j \in [m]} p_j^*\log \frac{p_j^*}{\tilde p_j^{(t)}} - \sum_{i\in[n]} B_i\log \tilde \nu_i^{(t)} + \sum_{i\in[n]} B_i\log B_i\\ 
    &= -\Phi(b^*) + \sum_{j \in [m]} p_j^*\log \frac{p_j^*}{\tilde p_j^{(t)}} - \sum_{i\in[n]} B_i\log \tilde \nu_i^{(t)}.\label{eqTele}
\end{align}
We now lower bound the second and third terms from the above individually as follows. Starting with the second term,
\begin{align}
 \sum_{j \in [m]} p_j^*\log \frac{p_j^*}{\tilde p_j^{(t)}} &=  \sum_{j \in [m]} p_j^*\log \frac{p_j^*}{\hat p_j^{(t)}/\sum_{j'\in[m]} \hat p_{j'}^{(t)}} + \sum_{j \in [m]} p_j^*\log \frac{\hat p_j^{(t)}}{\tilde p_j^{(t)}} - \sum_{j \in [m]} p_j^*\log \sum_{j'\in[m]} \hat p_{j'}^{(t)}\\    
 &=  D\left(p_j^*\middle\|\frac{\hat p_j^{(t)}}{\sum_{j'\in[m]} \hat p_{j'}^{(t)}}\right)+ \sum_{j \in [m]} p_j^*\log \frac{\hat p_j^{(t)}}{\tilde p_j^{(t)}} - \log \sum_{j\in[m]} \hat p_{j}^{(t)}\\   
  &\ge  D\left(p_j^*\middle\|\frac{\hat p_j^{(t)}}{\sum_{j'\in[m]} \hat p_{j'}^{(t)}}\right)+ \sum_{j \in [m]} p_j^*\log \frac{1}{1+\varepsilon_p} - \log \frac{1}{1-\varepsilon_\nu}\\  
  &\ge  D\left(p_j^*\middle\|\frac{\hat p_j^{(t)}}{\sum_{j'\in[m]} \hat p_{j'}^{(t)}}\right) - \varepsilon_p - 2\varepsilon_\nu \ge - \varepsilon_p - 2\varepsilon_\nu\label{eqPbound}
\end{align}
Moving on the the third term, 
\begin{align}
 - \sum_{i\in[n]} B_i\log \tilde \nu_i^{(t)} &=  - \sum_{i\in[n]} B_i\log \hat \nu_i^{(t)} - \sum_{i\in[n]} B_i\log \frac{\tilde \nu_i^{(t)}}{\hat \nu_i^{(t)}}\\
 &\ge - \sum_{i\in[n]} B_i\log \hat \nu_i^{(t)} - \sum_{i\in[n]} B_i\log (1+\varepsilon_\nu)\\
 &= - \sum_{i \in [n]} B_i \log {\sum_{j\in [m]}\frac{\hat p_j^{(t)}}{\tilde p_j^{(t)}}} \frac{v_{ij}\hat b_{ij}^{(t)}}{\hat p_j^{(t)}} - \log (1+\varepsilon_\nu)\\
 &\ge - \sum_{i \in [n]} B_i \log \sum_{j\in [m]}\frac{1}{1-\varepsilon_p} \frac{v_{ij}\hat b_{ij}^{(t)}}{\hat p_j^{(t)}} - \log (1+\varepsilon_\nu)\\
 &= - \sum_{i \in [n]} B_i \log \sum_{j\in [m]} \frac{v_{ij}\hat b_{ij}^{(t)}}{\hat p_j^{(t)}} + \log (1-\varepsilon_p) - \log (1+\varepsilon_\nu)\\
  &\ge \Phi(\hat b^{(t)})- 2\varepsilon_p - \varepsilon_\nu \label{eqEGest}
\end{align}
Hence, in total, we find that 
\begin{equation}
    \hat\Delta_t 
    = \sum_{i\in[n], j \in [m]} b_{ij}^*\log \frac{\hat b^{(t+1)}_{ij}}{\hat b^{(t)}_{ij}} \ge \Phi(\hat b^{(t)}) - \Phi(b^*) - 3\varepsilon_p - 3\varepsilon_\nu.
\end{equation}
Taking the telescoping sum of $\hat\Delta_t$, we see that
\begin{equation}
\sum_{t=0}^T \hat \Delta_t  = \sum_{i\in[n], j \in [m]} b_{ij}^*\log \frac{b_{ij}^*}{b^{(0)}_{ij}} - \sum_{i\in[n], j \in [m]} b_{ij}^*\log \frac{b_{ij}^*}{\hat b^{(t+1)}_{ij}} \ge \sum_{t=0}^T \left( \Phi(\hat b^{(t)}) - \Phi(b^*) - 3\varepsilon_p - 3\varepsilon_\nu\right).
\end{equation}
Taking the upper bound of $\sum_{t=0}^T \hat \Delta_t$, we obtain
\begin{align}
\sum_{t=0}^T \hat \Delta_t  &= \sum_{i\in[n], j \in [m]} b_{ij}^*\log \frac{b_{ij}^*}{b^{(0)}_{ij}} - \sum_{i\in[n], j \in [m]} b_{ij}^*\log \frac{\hat b_{ij}^*}{\hat b^{(t+1)}_{ij}}\\
&\le \begin{multlined}[t]\sum_{i\in[n], j \in [m]} b_{ij}^*\log \frac{b_{ij}^*}{b^{(0)}_{ij}} - \sum_{i\in[n], j \in [m]} b_{ij}^*\log \frac{b_{ij}^*}{\hat b^{(t+1)}_{ij}/\sum_{i'\in[n], j'\in [m]}\hat b^{(t+1)}_{i'j'}} \\+  \sum_{i\in[n], j \in [m]} b_{ij}^*\log \sum_{i'\in[n], j'\in [m]}\hat b^{(t+1)}_{i'j'}
\end{multlined}\\
&= D(b^*\|b^{(0)}) - D\left(b^*\middle\|\frac{\hat b^{(t+1)}}{\sum_{i'\in[n], j'\in [m]}\hat b^{(t+1)}_{i'j'}}\right) + \log \sum_{i'\in[n], j'\in [m]}\hat b^{(t+1)}_{i'j'}
\\
&\le  D(b^*\|b^{(0)}) - D\left(b^*\middle\|\frac{\hat b^{(t+1)}}{\sum_{i'\in[n], j'\in [m]}\hat b^{(t+1)}_{i'j'}}\right) + \log (1+\varepsilon_\nu)\\
&\le  D(b^*\|b^{(0)}) + \varepsilon_\nu\label{eqtelebound}
\end{align}

Hence, we obtain $\sum_{t=0}^T \left(\Phi(\hat b^{(t)}) - \Phi(b^*)\right) \le D(b^*\|b^{(0)}) + (3T+4) \varepsilon_\nu + (3T+3) \varepsilon_p$. Instead of $T$, we plug in $T-1$ to obtain
\begin{equation}
\sum_{t=0}^{T-1} \left(\Phi(\hat b^{(t)}) - \Phi(b^*)\right)  \le D(b^*\|b^{(0)}) + (3T+1) \varepsilon_\nu + 3T \varepsilon_p.
\end{equation}
With a simple observation that
\begin{equation}
T\cdot\min_{t\in[T]}\left(\Phi(\hat b^{(t)}) - \Phi(b^*)\right) \le  \sum_{t=0}^{T-1} \left(\Phi(\hat b^{(t)}) - \Phi(b^*)\right),
\end{equation}
we find 
\begin{equation}
\min_{t\in[T]}\left(\Phi(\hat b^{(t)}) - \Phi(b^*)\right) \le \frac{D(b^*\|b^{(0)})}{T} + 4 \varepsilon_\nu + 3 \varepsilon_p.
\end{equation}

From \cref{lemmaInit}, we know that $D(b^*\|b^{(0)})\le \log m$. Then by setting
$\varepsilon_{\nu} \le \frac{\log m}{8T}$ and $\varepsilon_{p} \le \frac{\log m}{6T}$, 
we obtain
\begin{equation}
    \min_{t\in[T]}\Phi(\hat b^{(t)}) - \Phi(b^*) \le \frac{2\log m}{T}.
\end{equation}
\end{proof}

Next, we prove \cref{theoremEsticonverge}.

\esticonverge*

\begin{proof}
We slightly modify the proof of \cref{theoremFPRconverge}, and note that by \cref{eqTele} and \cref{eqPbound}, we have
\begin{align}
\hat\Delta_t 
    &= -\Phi(b^*) + \sum_{j \in [m]} p_j^*\log \frac{p_j^*}{\tilde p_j^{(t)}} - \sum_{i\in[n]} B_i\log \tilde \nu_i^{(t)}
    &\ge -\Phi(b^*) - \varepsilon_p - 2\varepsilon_\nu - \sum_{i\in[n]} B_i\log \tilde \nu_i^{(t)}
\end{align}
Taking the telescoping sum and the upper bound from \cref{eqtelebound}, we obtain
\begin{equation}
\sum_{t=0}^{T-1} \left(- \sum_{i\in[n]} B_i\log \tilde \nu_i^{(t)}-\Phi(b^*)\right) \le \sum_{t=0}^{T-1}  \hat\Delta_t + 2T\varepsilon_\nu + T \varepsilon_p \le D(b^*\|b^{(0)}) + (2T+1)\varepsilon_\nu + T \varepsilon_p,
\end{equation}
where we can note 
\begin{equation}
\min_{t\in [T]} \left(- \sum_{i\in[n]} B_i\log \tilde \nu_i^{(t)} - \Phi(b^*)\right)  \le \frac{D(b^*\|b^{(0)})}{T} + 3 \varepsilon_\nu + \varepsilon_p.
\end{equation}
Let $t^* = \argmin_{t\in[T]} \left(- \sum_{i\in[n]} B_i\log \tilde \nu_i^{(t)} - \Phi(b^*)\right)$. Then by \cref{eqEGest}, we have the following:
\begin{equation}
 - \sum_{i\in[n]} B_i\log \tilde \nu_i^{(t)} \ge \Phi(\hat b^{(t)})- 2\varepsilon_p - \varepsilon_\nu
\end{equation}
Then we can obtain
\begin{equation}
\Phi(\hat b^{(t^*)}) - \Phi(b^*)  \le \frac{D(b^*\|b^{(0)})}{T} + 4 \varepsilon_\nu + 3 \varepsilon_p.
\end{equation}
Lastly by setting
$\varepsilon_{\nu} \le \frac{\log \nu}{8T}$ and $\varepsilon_{p} \le \frac{\log m}{6T}$, 
we obtain
\begin{equation}
    \Phi(\hat b^{(t^*)}) - \Phi(b^*) \le \frac{2\log m}{T}.
\end{equation}
\end{proof}

\section{Experimental and implementation details}
\label{appendixExp}
Our experiments are conducted on a single NVIDIA P100 GPU and written with the \texttt{PyTorch} library~\citep{paszke2019pytorch}. The optimal objective value is approximately computed by taking the results of the $1000$-th iteration of the PR dynamics.

For the projected gradient descent (PGD) algorithm, our implementation is unlike \citet{gao2020first}, whose task is based on the CEEI scenario where agents are given a unit of fake money and whose end goal is only the allocation. We require information on both the allocation $x$ and price $p$, hence our algorithm output should be the bids $b$. Therefore, instead of formulating the problem after the EG objective function, we mirror\footnote{Pun intended.} the PR dynamics in its equivalence to mirror descent~\citep{birnbaum2011distributed} on the Shmyrev objective function and perform PGD on the latter (see \cref{algoPGD}).

\begin{algorithm}
\caption{Projected Gradient Descent}
\label{algoPGD}
\KwIn{Budget $B$, Value $v$, Learning rate $\gamma$, Iterations $T$}
\KwOut{Bids $b$}
\DontPrintSemicolon
$b_{ij}^{(0)} = \frac{B_i}{m}$\;
\For{$t = 0$ to $T$}{
    $r_{ij}^{(t)} = b_{ij}^{(t)} - \gamma \cdot (1 - \log v_{ij}/p_j^{(t)})$ \tcp{Gradient step}
    \For{$i = 0$ to $n$}{
       $b_{i, *}^{(t+1)} = \operatorname{Proj}(r_{i, *}^{(t)} \to \{x \in \mathbb{R}_+^n, \sum_k x_k = B_i\})$ \;\tcp{Projection step onto a B\_i-simplex}
    }
}
\Return $b^{(T)}$
\end{algorithm}

We formulate the Shmyrev objective function into the following form to obtain convergence guarantees and the step size:
\begin{equation}
f(x) = h(Ax) + \langle q, x \rangle
\end{equation}
where $x \in \mathbb{R}^n, A\in \mathcal{M}_{d\times n}(\mathbb{R}), h: \mathbb{R}^d \to \mathbb{R}, q \in \mathbb{R}^n$. Considering a flattened vector of the bids $b$, we note that if 
\begin{equation}
A = n\left\{\begin{pmatrix}1\\ 1\\ \vdots\\ 1\end{pmatrix}\right.\otimes 
\underbrace{\begin{pmatrix}
1& 0& \cdots& 0\\
0& 1& \cdots& 0\\
\vdots& \vdots& \ddots& \vdots\\
0& 0& \cdots& 1
\end{pmatrix}}_m,\quad q = \begin{pmatrix}-\log v_{11} \\ -\log v_{12}\\ \vdots \\ -\log v_{mn}\end{pmatrix},\quad h(x) = \sum_i x_i \log x_i
\end{equation}
then $f = \Psi$. Then by Theorem 3 of \citep{gao2020first}, by setting a learning rate of $\gamma= 1/L\|A\|^2$, where $L = 1/\min_{j, t} p_j^{(t)}$, we get linear convergence. Note that $\|A\|^2 = n$. \citet{gao2020first} further provide a line search procedure to set the constant multiplier in the learning rate as well as provide sharper convergence guarantees, but as we only run for $16$ iterations, we do not perform the line search and fix the learning rate to the initial learning rate that \citet{gao2020first} use in their empirical studies, which is $1000 / L\|A\|^2$. 

For QAE, we set $M = \sqrt{T\sqrt{n}} /16 = 32$. We scale down $M$ by the constant factor of $16$ to save memory consumption on the GPU, as we simulate QAE by computing the full probability distribution over $[M]$. We compensate for the loss in accuracy of the estimation by employing the median-of-means estimator~\citep{nemirovsky1983problem}, where we take the median of $3$ estimators constructed from the mean of $7$ samples from the QAE subroutine. In addition, we perform the maximum finding classically by line search in our simulation as opposed to a randomized protocol in the quantum algorithm, whose effects are only in regard to ignoring the failure possibility.
\end{document}